%% file: template.tex
\documentclass{article}
\usepackage{amsthm}
\usepackage{amsmath}
\theoremstyle{definition}
\newtheorem{definition}{Definition}[section]
\newtheorem{theorem}{Theorem}[section]
\newtheorem{lemma}[theorem]{Lemma}
\usepackage{arxiv}
\usepackage[utf8]{inputenc} 
\usepackage[T1]{fontenc}    
\usepackage{hyperref}       
\usepackage{url}            
\usepackage{booktabs}       
\usepackage{amsfonts}       
\usepackage{bm}
\usepackage{nicefrac}       
\usepackage{microtype}      
\usepackage{lipsum}
\usepackage{tikz}
\tikzstyle{vertex}=[circle, draw, inner sep=0pt, minimum size=6pt]
\newcommand{\vertex}{\node[vertex]}

\usepackage[]{algorithm2e}
\title{Strategyproof Facility Location Mechanisms on Discrete Trees}

\author{
 Alina Filimonov \\
Technion - Israel Institute of Technology\\
  Haifa, Israel \\

   \And
 Reshef Meir \\
  Technion - Israel Institute of Technology\\
  Haifa, Israel \\

}

\begin{document}
\maketitle

\begin{abstract}
We address the problem of strategyproof~(SP) facility location mechanisms on discrete trees. Our main result is a full characterization of onto and SP mechanisms. In particular, we prove that when a single agent significantly affects the outcome, the trajectory of the facility is almost contained in the trajectory of the agent, and both move in the same direction along the common edges. We show tight relations of our characterization to previous results on \emph{discrete lines} and on \emph{continuous trees}. We then derive further implications of the main result for infinite discrete lines. 
\end{abstract}

\section{Introduction}
 In facility location problems, a central planner has to determine the location of a public facility that needs to serve a set of agents. Once the facility is located, each agent incurs some cost. Importantly, in non-cooperative settings, agents may have an incentive to misreport their locations to decrease their costs. One key objective of the planner that received much attention in the multiagent systems literature is to design a mechanism that incentivizes agents to report their true locations, i.e., mechanisms that are strategyproof~(SP).

  An $n$-agent facility location mechanism on a domain $ \mathbb{D}$ receives a profile of the agents' locations $a = (a_1,\ldots, a_n) \in  \mathbb{D}^n$ and outputs a location in $ \mathbb{D}$ depending on the profile. We refer to the agent's location as her peak.
  We say that a mechanism is SP if it is a weakly dominant strategy of every agent to report truthfully.
 
 The fundamental characterization result for strategyproof facility location was given by Moulin~\cite{moulin1980strategy}, who characterized the class of deterministic SP mechanisms on the real line when the preferences of the agents are single-peaked as ``generalized median voter schemes" (g.m.v.s.'s).
 An agent with single-peaked preferences on a line prefers a closer location to her peak over a distant location on the same side of her peak. 
 
  Border and Jordan~\cite{border1983straightforward}  proved that the characterization also applies for cases where the preferences are ``quadratic" (i.e., symmetric and single-peaked)---the more common model in facility location used in this work as well. An agent with quadratic preferences on a line prefers a closer location to her peak over a distant location
  
  As quadratic preferences are a special case of single-peaked preferences, the class of SP mechanisms for quadratic preferences may be larger. This is indeed the case e.g. for mechanisms on the discrete lines~\cite{dokow2012mechanism}, but not on continuous lines~\cite{border1983straightforward}.
 
  Schummer and Vohra~\cite{schummer2002strategy} generalized the result of Border and Jordan to prove that an SP mechanism on a continuous tree, under quadratic preferences, is a consistent collection of g.m.v.s.
 
 As hinted above, the trigger for the current work is the observation by Dokow et al.~\cite{dokow2012mechanism} that results on continuous graphs do not carry over to discrete graphs. In particular, while g.m.v.s entails that the trajectory of the facility is contained in the trajectory of the moving agent on a line (we later observe this also applies for continuous trees), Dokow et al. show it is only ``almost contained" when the line is made of discrete vertices.
 Similarly, while a strategyproof onto mechanism on a continuous circle must be dictatorial~\cite{schummer2002strategy}, it is only ``almost dictatorial" when the circle is discrete~\cite{dokow2012mechanism}. These extensions may be subtle, but they help us understand what in the characterization is inherent to the topology of the graph.  
 
 Given these previous results, a natural question is whether a similar extension can be applied to discrete trees.
  
  As we will later show, a na\"ive extension of the properties defined in~\cite{dokow2012mechanism} fails. We therefore formulate similar properties to characterize the valid moves of the facility under SP, onto mechanisms on discrete trees. In particular, we provide a definition of ``almost Pareto efficient" mechanisms that might be of independent interest. 
  
Recent research by Peters et al.~\cite{peters2019unanimous} provides a different characterization of randomized strategyproof voting mechanisms on trees and on other graphs (which, of course, include deterministic mechanisms), under general single-peaked preferences. However, the class of strategyproof mechanisms on discrete trees under single-peaked preferences is not equivalent to the one under quadratic preferences and therefore the characterization in~\cite{peters2019unanimous} does not apply to our study.

\subsection{Structure and Contribution}
After some preliminary notation in Section~\ref{chap:prelims}, we provide an alternative characterization of onto, SP mechanisms on continuous trees in Section~\ref{sec:cont trees}, based on the work of Schummer and Vohra~\cite{schummer2002strategy}.

 We then present our main result in Section~\ref{sec:main}---a full characterization of onto, SP mechanisms on discrete trees. In contrast to the work of Dokow et al.~\cite{dokow2012mechanism}, our proof also works for infinite trees (with bounded degree).  

In Section~\ref{sec:linesSI}, we derive a characterization of SP and shift-invariant mechanisms on infinite discrete lines.

\section{Related Work}
\label{chap:related}
Following the initial work of Black~\cite{black1948rationale}, several researchers have developed characterizations of deterministic and probabilistic strategyproof facility location mechanisms in various scenarios.

     Schummer and Vohra~\cite{schummer2002strategy}, beyond their work on trees, showed that any onto SP mechanism on the continuous cycle must be a dictatorship and that any SP mechanism on a graph has a dictator in a subdomain. Their work was extended to discrete cycles in \cite{dokow2012mechanism}. 

    Additional variations of the problem include the multiple facility problem~\cite{chen2020strategyproof,liu2020multiple}, the obnoxious facility location problem~\cite{cheng2011mechanisms,cheng2013strategy}, the heterogeneous facility location problem~\cite{Anastasiadis2018HeterogeneousFL, duan2019heterogeneous}, and the activity scheduling problem~\cite{xu2020strategyproof}. 
    
      Todo et al.~\cite{todo2011false} extended Moulin's work for characterizing the class of false-name-proof mechanisms on the continuous line. Their work was extended to discrete structures in~\cite{nehama2019manipulations,todo2019false}. The motivation for designing such mechanisms is to prevent agents from submitting multiple reports under different identities, e.g.,
in internet polls by creating different e-mail addresses. A later work of Wada et al.~\cite{wada2018facility} on variable and dynamic populations characterizes mechanisms that incentivize the agents to participate in the reporting process.

Finally, concrete cost functions also allow us to measure the \emph{social cost} (e.g., as the sum or max of agents' costs). The research line of \emph{approximate mechanism design without money} builds on characterizations such as those mentioned above, and seeks the mechanisms that minimize the social cost among all strategyproof mechanisms. Incidentally, the iconic domain for this line of work, as reflected in the fundamental paper of Procaccia and Tennenholtz~\cite{procaccia2009approximate}, is the facility location problem. Their work was extended to the domain of continuous graphs by Alon et al.~\cite{alon2010strategyproof}. 
     In the context of onto and strategyproof mechanisms on trees (either continuous or discrete), the question of minimizing the utilitarian social cost, defined as the sum of agents' costs, is moot since there is a simple mechanism for trees (the median voter) that is both strategyproof and socially optimal.

\section{Preliminaries}
\label{chap:prelims}

Consider an unweighted, undirected, bounded degree discrete tree $T = (V,E)$ with a set $V$ of vertices and a set $E$ of edges. The sets $V$ and $E$ can be infinite. For any two vertices $v_1,v_2 \in V$, $d(v_1, v_2)$ is the length of the unique path between $v_1$ and $v_2$. The distance between two sets of vertices $A \subseteq V$ and $B \subseteq V$ is the length of the shortest path between any pair of vertices $a,b$, where $a\in A$ and $b \in B$. We sometimes refer to a discrete tree as the set of its vertices. Consequently, the distance between two subtrees of a tree is the distance between the corresponding sets of vertices. For $u,w \in V$ with $u \neq w$, $[u,w]$ is the sequence of vertices $v_0,\ldots, v_k$ on the unique path of length $k$ between $u$ and $v$ s.t. $v_0 = u, v_k = w$. We denote by $(u,w]$ the sequence $v_1,\ldots v_k$, and by $(u,w)$ the sequence $v_1,\ldots v_{k-1}$, where $\{v_i,v_{i+1}\} \in E$ for all $i= 0,\ldots k-1$. We say that $e = \{u,w\} \in [a,b]$ if $[u,w] \in [a,b]$. A line-graph is a tree with a maximum degree of 2.

Let $N = \{1,\ldots,n\}$ be the set of agents, and $ a = (a_1,\ldots,a_n) \in V^n$ be a location profile, where $a_i \in V$
denotes the location of agent $i$ for every $i \in N$. The location profile of all agents excluding agent $i$ is denoted by $a_{-i} \in V^{n-1}$.
A deterministic facility location mechanism on a discrete tree is a
function $f : V^n \rightarrow V $, that maps a given profile of the agents' locations to a single location.

The notation $v_j \succ_i v_k$ indicates that agent $i$ prefers vertex $v_j$ over vertex $v_k$. The notation $v_j \succeq_i v_k$ indicates that agent $i$ prefers vertex $v_j$ over vertex $v_k$, or is indifferent between the two. 

In this research we assume the agents' costs are inversely related to their distance from the chosen location. We refer to such cost functions as ``quadratic" costs. The class of preferences induced by quadratic costs is single-peaked and symmetric. Formally, for every agent $i \in N$, 
$$\forall v_j,v_k \in V: v_j \succ_i v_k \iff d(a_i, v_j) < d(a_i, v_k) $$

Next, we give the standard definitions of mechanism properties:
          \begin{definition}[Strategyproof] A mechanism $f$ is \textbf{strategyproof} (SP) if an agent does not benefit from reporting a false location. Formally, $f$ is strategyproof if for every agent $i \in N$, every profile $a \in V^n$ and every alternative location $a'_i \in V$, it holds that  $$d(a_i, f(a_i, a_{-i})) \leq  d(a_i,f(a'_i, a_{-i})).$$
         \end{definition}
     \begin{definition}[Onto] A mechanism $f$ is \textbf{onto}, if for every location $x \in V$ there is a location profile $a\in V^n$ s.t. $f(a)=x$.
    \end{definition}
    \begin{definition}[Unanimous] A mechanism f is \textbf{unanimous} if for every location $x \in V$, $f(x, \ldots , x) = x$
     \end{definition}
     Clearly, every unanimous mechanism is onto. The following lemma provides a necessary condition for an onto, SP mechanism on any domain.
     \begin{lemma}[Barbera and Peleg~\cite{barbera1990strategy}]\label{claim:unanimous}
     Every mechanism that is both onto and SP, is unanimous.
     \end{lemma}
     
     The definitions above apply also for continuous trees. A finite continuous tree $G=(V,E)$ is a connected, acyclic collection of curves of finite length. $E$ is the set of curves and $V$ is the set of the extremities and intersections of the curves~\cite{schummer2002strategy}. Let $L\subseteq V$ denote the set of extremities only. For all $p_1,p_2 \in G$, $d(p_1,p_2)$ is the length of the unique path between $p_1$ and $p_2$. We denote by $(p_1,p_2)$ the open segment between $p_1$ and $p_2$, and by $[p_1,p_2]$ the closed segment between the two points. For any point $p$ and a set $S\subseteq G$ (which may itself be a segment), the notation $[p,S]$ stands for the segment $[p, s]$ where $s = \arg\min_{s \in S}d(p,s)$. For a mechanism on a continuous tree, the agents and the facility can be placed on arbitrary points on the edges. A mechanism on a continuous tree is therefore a function $f : G^n \rightarrow G$. The definitions are illustrated in Fig~\ref{fig:cont}.
     
\section{SP Mechanisms on Continuous Trees}\label{sec:cont trees}

Schummer and Vohra~\cite{schummer2002strategy} provided a characterization of onto, SP mechanisms on continuous trees. They showed that when the agents' preferences are quadratic, every SP, onto mechanism on the continuous tree is based on a set of generalized median voter schemes, defined in~\cite{moulin1980strategy}, satisfying a consistency condition.

In this section, we provide an alternative characterization of onto, SP mechanisms on continuous trees, that relies on previous works~\cite{border1983straightforward,schummer2002strategy}.

 \begin{figure}[t]
 \centering
 \includegraphics[scale=0.68]{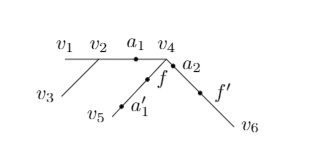}
    \caption{In this tree, $\{v_i| 1 \leq i \leq 6\}$ is the set of vertices. $\{v_1,v_3,v_5, v_6\}$ is the set of leafs. The agents are located at $a_1$ and $a_2$ (i.e., not on any vertex). The facility is located at point $f$, which is closer to $a_2$.}
    \label{fig:cont}
\end{figure}
\paragraph{Previous results on a Continuous Line}

The following definition given in~\cite{schummer2002strategy} describes onto and SP mechanisms on a continuous line. It is similar to the one introduced by Moulin~\cite{moulin1980strategy} for single-peaked preferences and confirmed for quadratic preferences by Border and Jordan~\cite{border1983straightforward}. 
\begin{definition}[Generalized Median Voter Scheme~\cite{schummer2002strategy}]
 A function $g^{xy}$ is called a \textbf{generalized median voter scheme} (g.m.v.s.) on $[x, y]$ if there exist $2^{|N|}$ points in $[x, y]$, $\{\alpha_S^{xy}\}_{S \subseteq N}$ such that:
\begin{enumerate}
    \item $S \subset R$ implies that $d(\alpha_S^{xy},x) \leq d(\alpha_R^{xy},x)$.
    \item $\alpha_{\emptyset}^{xy} = x$ and $\alpha_{N}^{xy} = y$.
    \item For all $a \in [x,y]^n$, $g^{xy}(a)$ is the unique point satisfying
    \begin{align}
        d(g^{xy}(a),x) = \textit{max}_{S\subset N}\textit{min}\{(d(a_i,x))_{i \in S}, d(\alpha_S^{xy},x)\}\notag
    \end{align}
\end{enumerate}
\end{definition}

The following is a key property of the class of g.m.v.s.'s defined in~\cite{border1983straightforward}. It implies that when an agent moves without crossing the mechanism outcome, the facility does not move. 

\begin{definition}[Uncompromising~\cite{border1983straightforward}]\label{def:unc}
A mechanism $f: \mathbb{R}^n \rightarrow \mathbb{R}$ is called \textbf{uncompromising} if for every $a \in \mathbb{R}^n$, $i \in N$, $a'_i \in \mathbb{R}$, it holds that:
\begin{enumerate}
    \item $a_i > f(a)$ implies that $f(a_{-i},a'_i) = f(a)$ for all $a'_i \geq f(a)$
    \item $a_i < f(a)$ implies that $f(a_{-i},a'_i) = f(a)$ for all $a'_i \leq f(a)$
\end{enumerate}
\end{definition}

\begin{lemma}[Border and Jordan~\cite{border1983straightforward}]\label{claim: unc lemma}
    Suppose that $f:\mathbb{R}^n \rightarrow \mathbb{R}$ is SP and unanimous. Then $f$ is uncompromising.
\end{lemma}
As shown in~\cite{schummer2002strategy}, this also applies for every mechanism $f: [x,y]^n \rightarrow [x,y]$ where $[x,y]$ is a finite interval in $\mathbb{R}$. Note that every SP and onto mechanism is unanimous by Lemma~\ref{claim:unanimous}. Therefore we can rely on uncompromisingness for our characterization of onto and SP mechanisms.
\paragraph{Previous results on a Continuous Tree}
\begin{definition}[Graph Restriction~\cite{schummer2002strategy}]
 For any subgraph $G' \subset G$, the
\textbf{graph restriction} of $f:G^n \rightarrow G$ to G' is the function $ f |_{G'}: G'^{n} \rightarrow G$ s.t for all profiles $a \in G'^{n},f|_{G'}(a) = f(a)$
\end{definition} 
By~\cite{schummer2002strategy}, if mechanism $f$ is SP and onto, then for every $a \in [x,y]^n$, $f |_{xy}( a)\in[x, y]$. 

The following property characterizes onto and SP mechanisms on a continuous trees. For all $x, y \in L$ and $a_i \in G$, let the
unique point in $[x, y]$ closest to $a_i$ be denoted $a_{i}|_{xy} =\textit{arg min}_{z \in [x, y]} d(z, a_i)$.
\begin{definition}[Extended Generalized Median Voter Scheme~\cite{schummer2002strategy}]
 A mechanism $f$ is an \textbf{extended generalized median voter scheme} (e.m.v.s.) if
 \begin{enumerate}
     \item For all $w, x, y, z \in G$, $f |_ {xy}$ and $f |_{wz}$ are consistent g.m.v.s.'s
     \item For all $a \in G^n, f( a)$ is the unique point p such that for all $x, y \in L$,
$p \in [x, y]$ implies $f | _{xy}( a|_{xy})=p$
 \end{enumerate}
\end{definition}
\begin{theorem}[Schummer and Vohra~\cite{schummer2002strategy}]\label{claim: emvs}
 For any continuous tree $G$, a rule $f$ is SP and onto if and
only if it is an e.m.v.s.   
\end{theorem}
We omit the definition of consistency since it is not relevant for our purpose.
\subsection{Our Characterization}
We rely on the characterization in~\cite{schummer2002strategy} to formulate our characterization for SP and onto mechanisms on continuous trees, which is a conceptual step on the way to our main result on discrete trees. The following properties limit the effect of an agent's move on the outcome of a mechanism on a continuous tree.
\begin{definition}[Tree Monotone]
 A mechanism $f$ on the discrete tree is \textbf{tree monotone}~(TMON) if for every profile $a \in G^n$, every agent $i\in N$, every location $a'_i \in G$ and every segment $[x,y]$ s.t. $[x,y] \subseteq [a_i,a'_i] \cap [f(a),f(a_{-i},a'_i)]$, it holds that 
 \begin{align}
     d(a_i, x) < d(a_i, y) \Leftrightarrow d(f(a), x) < d(f(a), y) \notag
 \end{align}
 Intuitively, TMON means that the facility moves in the same direction as the moving agent (if it crosses the agent's path at all).
\end{definition}
\begin{definition}[Trajectory contained] A mechanism on a tree is \textbf{Trajectory Contained} (TC) if for all $a, a'=(a_{-i},a'_i)$ it either  holds that $[f(a),f(a')]\subseteq [a_i,a'_i]$, or $f(a)=f(a')$.
\end{definition}
In words, either the trajectory of the outcome is contained in the trajectory of the agent, or the facility does not move at all. In Fig.~\ref{fig:cont}, when agent 1 moves from $a_1$ to $a'_1$, the facility moves from $f$ to $f'$. This violates TC since $[v_4,f'] \not\subseteq [a_i,a'_i]$. This also violates TMON since the facility and the agent move in opposite directions in the segment $[f,v_4]$.

\begin{lemma}\label{claim: TC1}
    Every onto and SP mechanism $f$ on the continuous tree is TC.
\end{lemma}
\begin{proof}
Consider an SP, onto mechanism $f:G^n \rightarrow G$ on a continuous tree. Assume by contradiction that there exists an agent $i$ and two profiles $a, a'=(a'_i,a_{-i})$ s.t. $ f(a) \neq f(a') $ and w.l.o.g., that $f(a) \notin [a_i,a'_i]$.
By Theorem~\ref{claim: emvs}, $f$ is an e.m.v.s. and therefore it is a collection of g.m.v.s's. Let $g^{xy}$ denote the g.m.v.s. on $[x,y]$, where $x,y \in L$ and $[f(a), f(a')] \subseteq [x,y]$. Note that $g^{xy} = f|_{xy} $. By Lemma~\ref{claim: unc lemma}, $g^{xy}$ is uncompromising.
From the second property of the e.m.v.s., it holds that
\begin{align}
    g^{xy}(a|_{xy}) = f(a) \textit{ and }g^{xy}(a'^{xy}) = f(a')\notag
\end{align}
We divide into two cases, according to the locations $a_i|_{xy},a'_{i}|_{xy}$. Note that for any other agent $j \neq i$, $a'_j=a_j$ and thus $a'_j|_{xy} = a_j|_{xy}$.
\begin{enumerate}
    \item $|[a_i,a_i|_{xy}] \cap [a'_i,a'_i|_{xy}]| > 0$: In this case, $a_i|_{xy} = a'_i|_{xy}$ and therefore  $ g^{xy}(a|_{xy}) = g^{xy}(a'|_{xy})$ contradicting the assumption that $f(a) \neq f(a')$.
     \item $|[a_i,a_i|_{xy}] \cap [a'_i,a'_i|_{xy}]| = 0$: In this case, the path between $a_i$ and $a'_i$ must intersect the segment $[x,y]$ and therefore,
     $[x,y] \cap [a_i,a'_i] = [a_i|_{xy},a'_i|_{xy}]$. Recall that by our initial assumption $$f(a) = g^{xy}(a|_{xy}) \notin [a_i,a'_i],$$ thus $f(a) \notin [a_i|_{xy},a'_i|_{xy}]$, contradicting the uncompromisingness of $g^{xy}$.\qedhere
\end{enumerate}
\end{proof}
\begin{lemma}\label{claim: TMON1}
    Every SP mechanism on the continuous tree is TMON.
\end{lemma}
The proof follows from the definitions of TMON and SP.

\begin{lemma}\label{claim: tc and tmon is SP}
    Every TC and TMON mechanism on the continuous tree is SP.
\end{lemma}
\begin{proof}
Assume by contradiction that there exists a TC, TMON mechanism $f$ that violates SP. Consider a pair of profiles $a,a' = (a_{-i},a'_i)$ s.t. $f(a) \neq f(a')$. 
Since $f$ is TC, it holds that $$[a_i,a'_i] \cap [f(a),f(a')] = [f(a),f(a')]$$ and therefore by TMON,  
$$d(a_i, f(a)) \leq d(a_i,f(a').$$
Therefore no agent can benefit from reporting a false location.
\end{proof}

\begin{theorem}\label{claim: TC theorem1}
    An onto mechanism on the continuous tree is SP if and only if it is TMON and TC.
\end{theorem}
\begin{proof}
The proof follows from Lemmas~\ref{claim: TC1},~\ref{claim: TMON1} and~\ref{claim: tc and tmon is SP}.
\end{proof}

\section{SP Mechanisms on Discrete Trees}\label{sec:main}
In this section we provide a complete characterization of onto, SP mechanisms on discrete trees, generalizing the result of Dokow et al. for discrete lines~\cite{dokow2012mechanism}.

Before presenting the main result, we show that a na\"ive extension of the properties defined for mechanisms on discrete lines in~\cite{dokow2012mechanism} fails for trees. Their result implies that an agent can affect the outcome of the mechanism only in a way in which its trajectory intersects the trajectory of the facility in at least two consecutive points.

The mechanism described in  Fig.~\ref{fig:01} is an example of an SP, onto mechanism that violates a na\"ive extension of this property.
Agent~1 is located at vertex $y=0$. Agent~2 is initially at $3$ and moves to $4$. As a result, the facility moves from vertex $1$ to $0$ without intersecting the segment $[3,4]$.
 
 \paragraph{Quadratic vs. single-peaked preferences} 
 \begin{definition}[Single-Peaked~\cite{peters2019unanimous}]
 A preference of an agent $i$ is \textbf{single-peaked} on a graph $G$ if there is a spanning tree $T = (V,E)$ of $G$ such that
for all distinct $x, y \in V$ with $a_i \neq y$,
$$x \in [a_i, y) \Rightarrow x \succ_i y $$
 \end{definition}
 The following example shows that under single-peaked preferences, the mechanism in Fig.~\ref{fig:01} is not SP. Assume the preferences of the agents are as follows:
 \begin{enumerate}
     \item Agent 1: $0 \succ_1 2 \succ_1 3 \succeq_1 1 \succ_1 4$
     \item Agent 2: $3 \succ_2 4 \succeq_2 2 \succ_2 0\succ_2 1$
 \end{enumerate}
 Both agents have single-peaked preferences according to the definition in~\cite{peters2019unanimous}. In particular, the preference of the first agent is quadratic. The preference of the second agent is not, since she strictly prefers vertex 0 over vertex 1. If both agents report truthfully, the facility will be located at vertex 1. However, if the second agent reports vertex 4 as her peak, the facility will be located at vertex 0 and the agent will benefit.  
 
 We conclude that similarly to the case of the line-tree, quadratic preferences allow more SP mechanisms than single-peaked preferences, and therefore the characterization of probabilistic SP mechanisms under single-peaked preferences in~\cite{peters2019unanimous} does not apply for quadratic preferences. 
\begin{figure}[t]
 \centering
\begin{minipage}{0.3\textwidth}
\input{alg}
\end{minipage}~~\begin{minipage}{0.3\textwidth}
\input{alg_graph}
\end{minipage}
    \caption{An example of a mechanism that is SP and onto on a discrete tree, that violates the properties defined by Dokow et al.~\cite{dokow2012mechanism}.}
    \label{fig:01}
\end{figure}
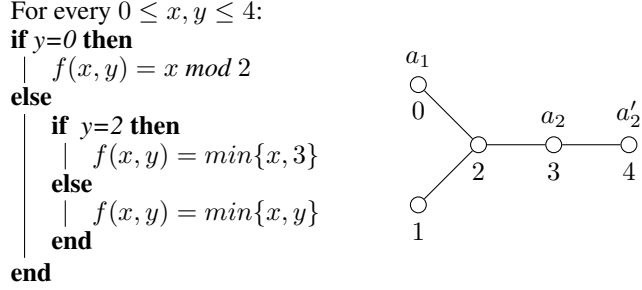

\subsection{Basic Mechanism Properties}\label{sec:basic}
Here we define several new terms which are specific for mechanisms on discrete trees.
 \begin{definition}[Tree]
$tree(a \rightarrow b, v)$ is the subtree which includes only $v$ and vertices which are accessible from $v$, via the edges that are not in $[a,b]$.
 \end{definition}

  \begin{definition}[Depth]
  $depth(a \rightarrow b,v)$ is the distance of vertex $v$ from $[a,b]$.
\end{definition}
We demonstrate the above definitions in Fig.~\ref{fig2}: $tree(a_1 \rightarrow a'_1,v_1)$ contains  $v_1$ (at depth 0) and another node at depth~1. $depth(a_1 \rightarrow a'_1,v_2)=2$ in the  subtree rooted by $a_1$.

 Our next definitions are intended to generalize the properties defined in \cite{dokow2012mechanism}. 
\begin{definition}[$m$-tree step independent]
A mechanism $f$ is \boldsymbol{$m$}\textbf{-tree step independent} ($m$-TSI) if for every $a \in V^n$, $i \in N$, $a'_i \in V$ s.t. $d([a_i , a'_i ], f(a)) > m$, it holds that 
\begin{align}
    tree(a_i \rightarrow a'_i,f(a)) =  tree(a_i \rightarrow a'_i,f(a_{-i},a'_i))\notag
\end{align}
\end{definition}
For $m=1$, the definition states that for every $a \in V^n$, $i \in N$, $a'_i \in V$ s.t.  
    \begin{align}
        |[f(a),f(a_{-i},a'_i)]\cap [a_i,a'_i]| \geq 2\notag
    \end{align}
    it holds that 
    $$d(f(a),[a_i,a'_i]) \leq 1$$
Fig.~\ref{fig2} illustrates a violation of the property. Mechanism $g$ violates $1$-TSI since 
$$tree(a_1 \rightarrow a'_1, g(a)) \neq  tree(a_1 \rightarrow a'_1, g(a'))\textit{ and } d([a_1,a'_1],g(a)) = 3 $$
\begin{definition}[Depth Balanced] A mechanism $f$ is \textbf{depth balanced} (DB) if for every $a \in V^n$, $i \in N$, $a'_i \in V$, it holds that 
 \begin{align}
     d(tree(a_i \rightarrow a'_i, f(a)), tree(a_i \rightarrow a'_i, f(a_{-i},a'_i))) \notag\geq \\ |depth( a_i \rightarrow a'_i,f(a)) - depth( a_i \rightarrow a'_i,f(a_{-i},a'_i))|\notag
 \end{align}
\end{definition}
Informally, DB means that when the facility moves as a result of a single agent's deviation, the distance between the tree of the original outcome and the tree of the new outcome is bigger than the difference between the depths of the outcomes. Fig.~\ref{fig2} illustrates a violation of the property by mechanism $g$. $g$ violates DB since
\begin{align*}
&d(tree(a_1 \rightarrow a'_1, g(a)), tree(a_1 \rightarrow a'_1,g(a'_1,a_2))) = 1 \notag\\
&<|depth(a_1 \rightarrow a'_1, g(a)) - depth(a_1 \rightarrow a'_1, g(a'_1,a_2)|=3
\end{align*}

\begin{definition}[Tree Pareto Location] Let $Int(a)$ be the set of interior vertices of the subtree defined by profile $a$: $$Int(a) = \{v \in V| \exists a_i,a_j \in a \textit{ s.t. }v \in (a_i,a_j)\}.$$ A location $x \in V$ is \textbf{tree Pareto} w.r.t. $a$ if $d(x,Int(a))\leq 1$ or $x = a_i$ for some $i \in N$.
\end{definition}
This definition generalizes the definition in~\cite{dokow2012mechanism} of a Pareto location on the discrete line. Note that it is weaker than the standard definition of Pareto.
\begin{definition}[Tree Pareto Mechanism] A mechanism $f$ is \textbf{tree Pareto}~(TPAR) if for every profile $a \in V^n$, $f(a)$ is a tree Pareto location w.r.t. $a$.
\end{definition}

 Mechanism $f$ in Fig.~\ref{fig2} violates TPAR w.r.t. profile $a'$ since $$d(f(a'), Int(a')) = d(f(a') , [v_1,v_2])=2$$
 
 \begin{definition}[Almost Depth Restricted] A mechanism $f$ is \textbf{almost depth restricted} (ADR) if for every $a \in V^n$, $i \in N$, $a'_i \in V$ s.t. $$f(a) \neq f(a'_i,a_{-i})\textit{ and } tree(a_i\rightarrow a'_i, f(a)) =  tree(a_i\rightarrow a'_i, f(a'_i,a_{-i})),$$ the following holds: Let $z$ be the unique point s.t. 
 $$z = [a_i,f(a)] \cap [a_i, f(a'_i,a_{-i})]\cap [f(a),f(a'_i,a_{-i})]$$ Then  
 \begin{enumerate}
     \item $d(f(a),z)= d(f(a'_i,a_{-i}),z)$
     \item $d(f(a),z) =  1$
 \end{enumerate}
\end{definition}
Informally, ADR means that when the facility moves as a result of a single agent's deviation, without intersecting the trajectory of the agent in at least two points, the new outcome has the same parent as the original outcome in the tree induced by the deviation. That is, either the facility does not move, or it moves to a sibling node. We can think of this property as ``approximate uncompromising" (replacing `1' with `0' in the definition would yield exact uncompromising). 
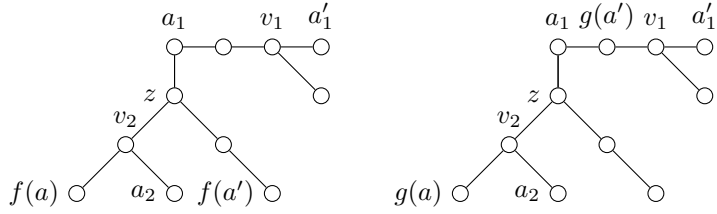
\begin{figure}[t]
 \centering
 \begin{minipage}{0.3\textwidth}
\input{double}
\end{minipage}~~\begin{minipage}{0.3\textwidth}
\input{doubleg}
\end{minipage}
\caption {Mechanism $f$ violates ADR and TPAR w.r.t. profile $a' = (a'_1,a_2)$, but satisfies TPAR w.r.t. profile $ a$.
Mechanism $g$ violates 1-TSI and DB.
}
\label{fig2}
\end{figure}
In Fig.~\ref{fig2}, ADR is violated by mechanism $f$ for the pair of profiles $(a,a')$, which differ by the location of agent $1$, since $d(f(a),z) = 2$.

\begin{definition}[Almost Trajectory Contained] A mechanism $f$ is \textbf{almost trajectory contained}~(ATC) if it is ADR and 1-TSI.
\end{definition}

  Finally, we use the term TMON defined in Section~\ref{sec:cont trees} to describe mechanisms on discrete trees as well.
  
  We now go on to characterize SP and onto mechanisms on discrete trees.
  \begin{theorem}\label{thm}
An onto mechanism $f$ on the discrete tree is SP if and only if it is TMON and ATC.
\end{theorem}
  First we show a weak characterization. We then use it to prove the tree Pareto property and the main property of ``almost trajectory containment", which consists of two properties, ADR and 1-TSI. To prove ADR, we first show that when an agent moves towards the facility, either the facility remains in place (as in the continuous case), or it moves exactly one step towards the agent and one step away. To prove 1-TSI, we first show that when an agent moves to a neighboring vertex on edge $e$, the facility can intersect $e$ only if it is at most one step away from $e$.
  \paragraph{DB and TMON}

\begin{lemma}\label{claim:proposal}
 A pair of profiles violates SP if and only if it violates DB \textbf{or} TMON.
\end{lemma}
\begin{proof}
\noindent``$\Rightarrow$"  Every pair of profiles $a,a=(a_{-i},a'_i)$ which violates SP, violates DB or TMON.\newline
   Consider a mechanism $f$ and a pair of profiles $a,a=(a_{-i},a'_i)$
 s.t. $$x:=f(a); x':= f(a_{-i},a'_i)  ;\text{ and }  d(a_i, x) >  d(a_i, x') $$
      For any mechanism $f$ and profile $a$ it holds that
      \begin{align}
          d(a_i, x) =  d(a_i,  tree(a_i \rightarrow a'_i, x)) +  depth( a_i \rightarrow a'_i,x)\label{eq:borg}
      \end{align}
 Assume that $f$ is TMON. Then it follows that:
 \begin{align}
  d(a_i, x') &=  d(a_i,  tree(a_i \rightarrow a'_i, x))\label{eq:mon2}\\ &+ d(tree(a_i \rightarrow a'_i, x),  tree(a_i \rightarrow a'_i, x'))\notag\\ &+ depth( a_i \rightarrow a'_i,x')  \notag
 \end{align}
 And thus from Eqs.~\eqref{eq:borg},~\eqref{eq:mon2} above,
$$d(tree(a_i \rightarrow a'_i, x),  tree(a_i \rightarrow a'_i, x')) + depth( a_i \rightarrow a'_i,x') <  depth( a_i \rightarrow a'_i,x)$$
And therefore,
$$
d(tree(a_i \rightarrow a'_i, x),  tree(a_i \rightarrow a'_i, x'))<  depth( a_i \rightarrow a'_i,x) - depth( a_i \rightarrow a'_i,x'),$$
contradicting DB.

\noindent``$\Leftarrow$" Every pair of profiles $a,a'=(a_{-i},a'_i)$ which violates TMON or DB, violates SP.
  \begin{description}
\item{TMON.} Consider a mechanism $f$ and a pair of profiles $a,a=(a_{-i},a'_i)$ that violates TMON, i.e.,
$\exists i \in N,a, a'_i$   s.t.
$$x:=f(a) \in  \textit{tree}_k(a_i \rightarrow a'_i); x':= f(a_{-i},a'_i) \in  \textit{tree}_j(a_i \rightarrow a'_i); \text{ and }  j<k$$ It follows that
\begin{align}
    d(a_i, x') =&  d(a_i,  tree(a_i \rightarrow a'_i, x')) +  depth( a_i \rightarrow a'_i,x') \label{eq:DR1}\\
    d(a_i, x) =&  d(a_i,  tree(a_i \rightarrow a'_i, x'))\label{eq:DR2} \\
    &+  d(tree(a_i \rightarrow a'_i, x'),  tree(a_i \rightarrow a'_i, x)) \notag\\
    &+  depth( a_i \rightarrow a'_i,x).\notag
\end{align}
Assume by contradiction that $f$ is SP. Then it follows that $\forall i,a, a'_i, d(a_i, x) \leq  d(a_i, x')$, and thus from Eqs.~\eqref{eq:DR1},\eqref{eq:DR2} above,
$$d(tree(a_i \rightarrow a'_i, x'),  tree(a_i \rightarrow a'_i, x)) +  depth( a_i \rightarrow a'_i,x) \leq  depth( a_i \rightarrow a'_i,x').$$
From the contradiction assumption $x$ and $x'$ are in different trees. Therefore,
\begin{equation}
    depth( a_i \rightarrow a'_i,x) <  depth( a_i \rightarrow a'_i,x').\label{eq:DR5}
\end{equation}
Symmetrically,  $\forall i,a, a'_i:  d(a'_i, x) \geq d(a'_i, x')$ from which we derive
$$depth(a_i \rightarrow a'_i,x) >  depth( a_i \rightarrow a'_i,x'),$$ contradicting Eq.~\eqref{eq:DR5} above.
Therefore the pair $a,a'$ violates SP.
\item{DB.} Consider an SP mechanism $f$ and a pair of profiles $a,a=(a_{-i},a'_i)$. For any mechanism $f$ and profile $a$ it holds that
\begin{align}
d(a_i, x) =  d(a_i, tree(a_i \rightarrow a'_i, x)) +  depth( a_i \rightarrow a'_i,x).\label{eq:dists}
\end{align}
From TMON:
\begin{align}
    d(a_i, x') =&  d(a_i,  tree(a_i \rightarrow a'_i, x))\label{eq:monb}\\ &+   d(tree(a_i \rightarrow a'_i, x), tree(a_i \rightarrow a'_i, x'))\notag \\&+ depth( a_i \rightarrow a'_i,x').\notag
\end{align}
From strategyproofness, $\forall i,a, a'_i, d(a_i, x) \leq  d(a_i, x')$, and thus from Eqs.~\eqref{eq:dists},\eqref{eq:monb} above,
 $$depth( a_i \rightarrow a'_i,x) \leq  d(tree(a_i \rightarrow a'_i, x),  tree(a_i \rightarrow a'_i, x')) +  depth( a_i \rightarrow a'_i,x')$$
Therefore,
\begin{align}
d(tree(a_i \rightarrow a'_i, x),  tree(a_i \rightarrow a'_i, x')) \geq  depth( a_i \rightarrow a'_i,x) -  depth( a_i \rightarrow a'_i,x')\label{eq:dr1}
\end{align}
    Symmetrically, from strategyproofness, $\forall i,a, a'_i:  d(a'_i, x) \geq  d(a'_i, x') $ from which we derive
    \begin{align}
    d(tree(a_i \rightarrow a'_i, x), tree(a_i \rightarrow a'_i, x')) \geq  depth( a_i \rightarrow a'_i,x') -  depth( a_i \rightarrow a'_i,x)\label{eq:dr2}
    \end{align}
    And thus, from Eqs.~\eqref{eq:dr1},\eqref{eq:dr2} above,
    $$d(tree(a_i \rightarrow a'_i, x),  tree(a_i \rightarrow a'_i, x')) \geq |depth( a_i \rightarrow a'_i,x') -  depth(a_i \rightarrow a'_i,x)|$$
    Therefore, any pair of profiles which violates DB and satisfies TMON, violates SP.
   \end{description}
\end{proof}  

Lemma~\ref{claim:proposal} characterizes all pairs of profiles that violate SP. We later show a stronger characterization that describes the pairs that indicate that an onto mechanism is not SP. These pairs differ only in a single agent's location, who does not necessarily benefit from misreporting, contrary to the characterization in Lemma~\ref{claim:proposal}.

  \paragraph{TPAR}

\begin{lemma}\label{claim:TP}
Every SP and onto mechanism $f$ on the discrete tree is TPAR.
\end{lemma}
\begin{proof}
Assume by contradiction that an onto, SP mechanism $f$ is not TPAR. Then there exists a profile $a$ s.t. $d(f(a),Int(a)) > 1$. Let $v$ denote the first vertex on the path from $f(a)$ to all other locations of $a$. Let $z^i := (a_1 \ldots , a_i, v,\ldots,v)$. Note that $z^n$ is the profile $a$ and $z^0$ is the profile $(v,\ldots v)$. By definition of $v$, we have that $f(a) = f(z^n) \in tree(a_i \rightarrow v, v)$.
 Since TPAR is violated, $v$ has exactly one neighbor $u$ on the path to any other agent location $a_i \neq v$. When moving from profile $z^{i+1}$ to $z^{i}$, it follows from TMON that $\forall 0\leq i \leq n-1$, 
 \begin{align}
     f(z^i) \in tree(a_i \rightarrow v, f(z^{i+1})) = tree(a_i \rightarrow v, v)\label{same tree}
 \end{align}
 From DB and~\eqref{same tree} we have that
\begin{align}
    depth( a_i \rightarrow v,f(z^i)) = 1 \neq 0 =  depth( a_i \rightarrow v,v)\label{depth}
\end{align}
 It follows from~\eqref{depth} that
 $$f(z^0) = f(v,..,v) \neq v$$ in contradiction to unanimity.
\end{proof}

  \paragraph{ADR}
Here we prove that ADR is a necessary condition for an onto, SP mechanism. We first show that if a pair of profiles violates the property, there exists a pair of profiles that violates the property in which the agent moves to a vertex on the path between the two outcomes.
\begin{lemma}\label{1-DR claim}
If an onto, SP mechanism $f$ violates ADR, there exists
a pair of violating profiles $a,a'=(a_{-i},a'_i)$ s.t. $$z= [a_i,f(a)] \cap [a_i, f(a')]\cap [f(a),f(a')] = a'_i$$ and  $$d(f(a),z)= d(f(a'),z)$$
\end{lemma}
\begin{proof}
 Assume by contradiction that there is a violating pair $(a,a')$ for which $f(a)$ and $f(a')$ are different vertices in the same tree w.r.t. the move $a_i \rightarrow a'_i$ and $a'_i \neq z$. Consider the profile $c$ where $c_{-i}=a_{-i}$ and $c_i=z$. If $f(c) = f(a)$ or $f(c) = f(a')$ the proof follows.
 
 Assume that $f(c) \neq f(a)$ and $f(c)\neq f(a')$. Then it follows from TMON that
 \begin{align}
     f(c) \in  tree (a_i \rightarrow a'_i,f(a)) ;\text{ and } tree (a_i \rightarrow a'_i,f(a))= tree(a_i \rightarrow a'_i,f(a'))\notag
 \end{align}
  Otherwise, if the tree that contains $f(c)$ is closer to $a_i$ than the tree that contains $f(a)$ and $f(a')$, mechanism $f$ is not TMON w.r.t. the pair of profiles $(a,c)$. If the tree that contains $f(c)$ is closer to $a'_i$, $f$ is not TMON w.r.t. the pair of profiles $(a',c)$. For the same reason, $f(c)$ belongs to the subtree of $z$, i.e. $d(f(c),[a_i,a'_i]) > d(z,[a_i,a'_i])$.
 
 Let $y$ denote the unique point s.t. 
 $y= [a_i,f(a)] \cap [a_i, f(c)]\cap [f(a),f(c)]$ and $y'$ the unique point s.t. $y'= [a_i,f(a')] \cap [a_i, f(c)]\cap [f(a'),f(c)]$ (see example in Fig.~\ref{fig3}).
 From DB, we have that $d(f(a),y)= d(f(c),y)$ and  $d(f(a'),y')= d(f(c),y')$.
 
 If $d(f(c),y)= d(f(a),y) = 1$ and $d(f(c),y')= d(f(a'),y') = 1$, it follows that $y$ and $y'$ are the same vertex, in contradiction to the assumption that the pair $(a,a')$ violates ADR. If $d(f(a),y) > 1$, the proof follows for the pair $(a,c)$. Otherwise, $d(f(a'),y') > 1$ and the proof follows for the pair $(a',c)$. 
 \end{proof}

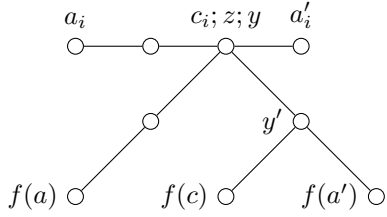
\begin{figure}[t]
 \centering
 \input{adr_suppl}
    \caption{An illustration of the proof of Lemma~\ref{1-DR claim}. Here the  violating pair of profiles is $(a,c)$. Note that this pair of profiles does not violate SP.}
    \label{fig3}
\end{figure}

The following lemma proves the necessity of ADR by showing that there is no violating pair of profiles in which the agent moves to a location on the path between the outcomes.
\begin{lemma}\label{1-DR}
Every onto SP mechanism on the discrete tree is ADR.
\end{lemma}
    \begin{proof}
    Assume that there exists an agent $i$ and two profiles $a, a'=(a'_i,a_{-i})$ s.t. $$x:=f(a); x':= f(a');\textit{ and } tree(a_i\rightarrow a'_i, x) =  tree(a_i\rightarrow a'_i, x')$$ 
    From DB, we have that $d(x,z)= d(x',z)$, 
    where $z$ is the unique point s.t. $z= [a_i,x] \cap [a_i, x']\cap [x,x']$.
    We show that ADR holds when $z=a'_i$. Assume by contradiction that 
    \begin{align}
        z=a'_i \textit{ and } d(x,z)= d(x',z) > 1\label{ADR cont} 
    \end{align}
    Among the violating pairs, let $(a,a'=(a_{-i},a'_i))$ be one of the pairs that minimize $\Sigma_{k\neq i}d(a_k,[a_i,a'_i])$. We let $v_1$ denote the first vertex on the path from $x$ to $z$ and $v_2$ denote the second vertex on the path from $x$ to $z$ (see Figs.~\ref{fig:adr}.1,~\ref{fig:adr}.2). Note that $v_2$ and $z$ is the same vertex if $d(x,z)=$ $d(x',z) = 2$.
   
      Since $f$ is TPAR, there must be some other agent $j$ s.t. $x$ is sufficiently close to the path from agent $j$ to agent $i$. Formally,
      $$\exists j: a_j \in tree(a_i \rightarrow v_1,v_1)\setminus v_1$$
       (see locations $a_{j_1},a_{j_2},a_{j_3}$ in  Fig.~\ref{fig:adr}.1). We define two profiles $b=(a_{-j},b_j=v_1)$ and $b'=(a'_{-j},b'_j=v_1)$, which differ from $a$ and $a'$ only by the location of agent $j$. Let $y$ denote $f(b)$ and $y'$ denote $f(b')$.

        For the pair of profiles $(a,b)$: If $x$ is on the path from $a_j$ to $v_1$ (location $a_{j_3}$ in Fig.~\ref{fig:adr}.1), it follows from DB and TMON that the facility will stay in the same tree w.r.t. the move $a_j \rightarrow v_1$ at depth 0 (location $y_3$ in Fig.~\ref{fig:adr}.3) or move to $tree(a_j \rightarrow v_1,v_1)$ and be located at depth 0 (location $y_2$ in Fig.~\ref{fig:adr}.3) or 1 (locations $y_1,y_4$ in Fig.~\ref{fig:adr}.3). Otherwise, if $x$ is not on the path from $a_j$ to $v_1$ (locations $a_{j_1}$, $a_{j_2}$ in Fig.~\ref{fig:adr}.1), $x$ will stay in the same tree at the same depth w.r.t. the move $a_j \rightarrow v_1$ (locations $y_1,y_3$ in Fig.~\ref{fig:adr}.3). Overall, the possible locations of $y$ are $x,v_1,v_2$ and the direct children of $v_1$. 
        
        Location $y'$ satisfies: 
            $d(y',v_1) \leq d(x',v_1)$.  Otherwise, SP is violated for the pair of profiles $(a',b')$. The location $y'$ also satisfies 
        \begin{align}
            y' \in tree(a_i\rightarrow a'_i, a'_i); \textit{ and } depth(a_i\rightarrow a'_i, y') = depth(a_i\rightarrow a'_i, y)\label{adr depth}
        \end{align}
        Otherwise, SP is violated for the pair of profiles $(b,b')$ (therefore $y'_1$ in Fig.~\ref{fig:adr}.4 is not a valid location). We divide into two cases by the possible locations of $y'$:
        \begin{enumerate}
         \item $y'$ is a child of $v_1$ in the $tree(a_i \rightarrow a'_i,v_1)$:
        In this case the pair $(a',b')$ is a violation of SP, since 
        \begin{align}
            d(a_j,y') \leq d(a_j,v_1)+d(v_1,y')\label{11}\\
            =d(a_j,v_1) + 1 \label{22}\\
            < d(a_j,v_1)+3  \notag\\
            \leq d(a_j,x')\label{33}
        \end{align}
        Eq.~\eqref{11} follows from the triangle inequality, Eq.~\eqref{22} follows from the case condition and Eq.~\eqref{33} follows from Eq.~\eqref{ADR cont}.
            \item $y'$ is not a child of $v_1$ in the $tree(a_i \rightarrow a'_i,v_1)$: 
            It follows from Eq.~\eqref{adr depth} and the case condition that 
            \begin{align}
                y' \notin tree(a_i \rightarrow a'_i,v_1)\setminus \{v_1\}\label{111}
            \end{align}
             Therefore, $v_1 \in [a_j, y']$. 
         From Eq.~\eqref{111},
         \begin{align}
             d(v_1,y')=d(v_1,x')\label{2222}
         \end{align}
         Otherwise, the pair $(a',b')$ violates SP, since the nearest location to $v_1$ between $y'$ and $x'$ is strictly closer to agent $j$ whether agent $j$ is located at $a_j$ or at $v_1$.
         From Eq.~\eqref{adr depth},
         \begin{align}
             depth(a_i \rightarrow a'_i, y') \in \{ depth(a_i \rightarrow a'_i, x), depth(a_i \rightarrow a'_i,v_1), depth(a_i \rightarrow a'_i, v_2) \}
         \end{align}
         We divide into two cases:
         \begin{enumerate}
             \item $depth(a_i \rightarrow a'_i, y') \in \{depth(a_i \rightarrow a'_i, v_1), depth(a_i \rightarrow a'_i, v_2)\}$ (locations $y'_4,y'_5$ in Fig.~\ref{fig:adr}.4):\newline
             \begin{align}
                 d(v_1,y') \leq d(v_1,z)+d(z,y')\label{333}\\
                  < d(v_1,z)+d(z,x')\label{444}\\
                   = d(v_1,x')\notag
             \end{align}
             contradicting Eq.~\eqref{2222}. Eq.~\eqref{333} follows from the triangle inequality and Eq.~\eqref{444} follows from the case condition.
               \item $depth(a_i \rightarrow a'_i, y')= depth(a_i \rightarrow a'_i, x)$ (locations $y'_2,y'_3$ in Fig.~\ref{fig:adr}.4):\newline
               Since $y'$ is not a direct child of $v_1$, the pair $(b,b')$ violates ADR in contradiction to the minimality of $\Sigma_{k\neq i}d(a_k,[a_i,a'_i])$, since agent $j$ is closer to the trajectory of agent $i$ in profiles $b,b'$ than in profiles $a,a'$. 
         \end{enumerate} 
        \end{enumerate}

      We have shown that there is no valid location for $y'$, and therefore ADR is not violated for the case in which $a'_i = z$. From Lemma \ref{1-DR claim}, every SP onto mechanism is ADR.
 \end{proof} 

 \begin{figure}[t]
   \centering
    \input{adr}
     \caption{ An illustration of the possible locations of agent $j$ and the facility locations for the profiles $a,a',b,b'$ in the proof of Lemma~\ref{1-DR}. }
     \label{fig:adr}
 \end{figure}
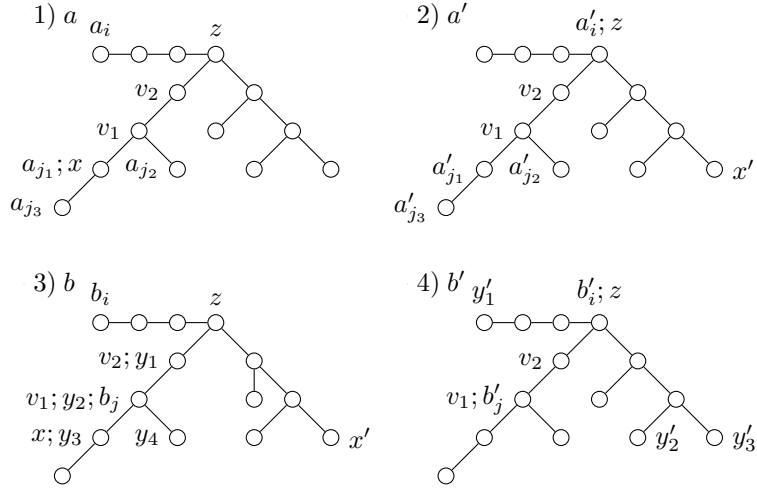
   
  \paragraph{1-TSI}
Here we prove that 1-TSI is a necessary condition for an onto, SP mechanism.
We first show that if there is a violation of the 1-TSI property, then w.l.o.g. it occurs when an agent moves a single step.
\begin{lemma}\label{onestep}
If an onto, SP mechanism $f$ violates 1-TSI, there exists
a pair of violating profiles $a,a'=(a_{-i},a'_i)$ s.t. $d(a_i,a'_i)=1$.
\end{lemma}
\begin{proof}
 Assume that $d(a_i,a'_i)>1$ and assume w.l.o.g. that $$d(f(a),[a_i,a'_i]) > 1$$
We denote by $x,y$ the two adjacent vertices on the path from $a_i$ to $a'_i$ s.t.
\begin{align}
   x'=(a_{-i},x);y'=(a_{-i},y);\notag\textit{ and } \\tree(a_i \rightarrow a'_i,f(x')) \neq  tree(a_i \rightarrow a'_i,f(y'))\label{eq: onestep different}
\end{align}
If there is more than one such pair then we select the pair $x,y$ closest to $a_i$ (see Fig.~\ref{fig onestep}). From Lemma~\ref{claim:proposal}, every SP, onto mechanism satisfies DB. Therefore, it holds that for every move $a_i \rightarrow z$ where $z \in [a_i,x]$, the facility stays in the same depth w.r.t. its initial tree, i.e., 
\begin{align}
    depth(a_i \rightarrow z,f(a_{-i},z)) =  depth(a_i \rightarrow z,f(a))\label{eq: onestep samedepth}
\end{align}
Assume by contradiction that
\begin{align}
    tree(x \rightarrow y,f(x')) =  tree(x \rightarrow y,f(y'))\label{onestep same tree}
\end{align} 
From Eq.~\eqref{eq: onestep different},
when agent $i$ moves from $x$ to $y$, the trajectory of the facility intersects the segment $[a_i,a'_i]$ in two points, and since d $(f(a),[a_i,a'_i]) \geq 2$, it holds that
\begin{align}
    d(f(a), f(y')) \geq 3\label{eq:3}
\end{align}
On the other hand, from Eq.~\eqref{onestep same tree} and ADR we have that
\begin{align}
    d (f(a),f(y')) \in \{0,2\} \label{eq:02}
\end{align}
contradicting Eq.~\eqref{eq:3} (see locations $f(x'),f(y')$ in Fig.~\ref{fig onestep}). Thus, it follows that 
$tree(x \rightarrow y,f(x')) \neq  tree(x \rightarrow y,f(y'))$.
Therefore, the pair $(x',y')$ violates 1-TSI by a one-step deviation, since $d(f(x'),[a_i,a'_i])>1$, and in particular, $d(f(a),[x,y])>1$
as required.
\end{proof}
  \begin{figure}[t] 
  \centering
 \input{1-TSI-one-step}
    \caption{An illustration of the proof of Lemma \ref{onestep}. The pair $(x'=(a_{-i},x),y'=(a_{-i},y))$ violates ADR.}
    \label{fig onestep}
\end{figure}
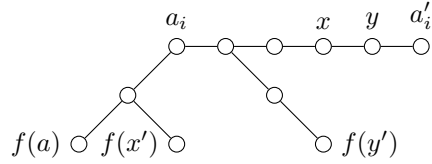
The following lemma proves the necessity of 1-TSI by showing that there is no violating pair of profiles in which the agent moves to a neighboring vertex.
\begin{lemma}\label{claim:1-TSI}
Every onto, SP mechanism on the discrete tree is 1-TSI.
\end{lemma}
\begin{proof}
Assume by contradiction that there exists a pair of profiles $a,a'=(a_{-i},a'_i)$ s.t. 
\begin{align}
    tree(a_i\rightarrow a', f(a)) \neq  tree(a_i\rightarrow a',f(a_{-i},a'_i))
\end{align} 
Assume w.l.o.g. that $d(f(a'),[a_i.a'_i])>1$. We can assume that $d(a_i,a'_i)=1$ from lemma \ref{onestep}. Among these pairs, let $(a,a')$ be the pair that minimizes $\Sigma_{k\neq i}d(a_k,[a_i,a'_i])$.
From TMON, we have that 
\begin{align}
    f(a) \in tree(a_{i}\rightarrow a'_i, a_i)\textit{ and }f(a') \in tree(a_{i}\rightarrow a'_i, a'_i)\notag
\end{align} 
Since $d(a_i,a'_i)=1$, $$d(tree(a_{i}\rightarrow a'_i, f(a)), tree(a_{i}\rightarrow a'_i, f(a'))) =1$$ and in order to satisfy the DB condition, the difference between the depths of $f(a)$ and $f(a')$ has to be at most 1. Therefore,
\begin{align}
    depth( a_{i}\rightarrow a'_i,f(a))\geq 1\notag
\end{align}
Let $p$ denote the first vertex on the path from $f(a)$ to $a_i$. From TPAR, there exists an agent $j$ s.t. $ a_j \in  tree(a_i\rightarrow a'_i,p)\setminus p$
(see locations $a_{j_1},a_{j_2},a_{j_3}$ in Fig.~\ref{fig:TSI}.1).

We define two profiles $b=(a_{-j},p)$ and $b'=(a_{-\{i,j\}},b'_i=a'_i,b'_j=p)$.

For the pair $(a,b)$, when agent $j$ moves to $p$, it follows from DB that the facility can only move to location $p$, or to a location $z$ s.t. 
\begin{align}
    d(p,z)=1\label{TSI z}
\end{align}

For the pair $(a,a')$, when agent $j$ moves to $p$, it follows from ADR that the facility can only move from $f(a')$ to a different child of $p$ in the tree induced by the move $a_i \rightarrow a'_i$ (see Fig.~\ref{fig:TSI}.2, Fig.~\ref{fig:TSI}.4). Therefore it holds that
\begin{align}
   &y:=f(b);y':=f(b'); y' \in  tree(a \rightarrow a_i',a'_i);\notag\\&depth(a \rightarrow a',y') =  depth(a \rightarrow a',f(a'))\geq 2\label{TSI 2}
\end{align} 
(See Fig.~\ref{fig:TSI}).
We show that every possible location of the facility for the profile $(a_{-j},p)$ violates SP or the minimality condition:
\begin{enumerate}
    \item $y \in  tree(a_{i}\rightarrow a'_i, a'_i)$ (location $y_2$ in Fig.~\ref{fig:TSI}.3): DB is violated for the pair $(b,b')$ since $y,y'$ are in the same tree w.r.t. the move $a_i \rightarrow a'_i$, but from Eq.~\eqref{TSI z}, $depth(a_{i}\rightarrow a'_i,y) =0$ while $depth(a_{i}\rightarrow a'_i,y') =2$ from Eq.~\eqref{TSI 2}.
    \item  $y \in  tree(a_{i}\rightarrow a'_i, a_i)$ (locations $y_1,y_3,y_4$ in Fig.~\ref{fig:TSI}.3): The pair $(b,b')$ violates 1-TSI. This contradicts the minimality of $\Sigma_{k\neq i}d(a_k,[a_i,a'_i])$, since agent $j$ is closer to the trajectory of agent $i$ in profiles $b,b'$ than in profiles $a,a'$. 
\end{enumerate}
\end{proof}
   \begin{figure}[t]
       \centering
         \input{1-TSI}
       \caption{An illustration of the possible locations of agent $j$ and the facility locations $y,y'$ for the profiles $a,a',b,b'$.}
       \label{fig:TSI}
   \end{figure}
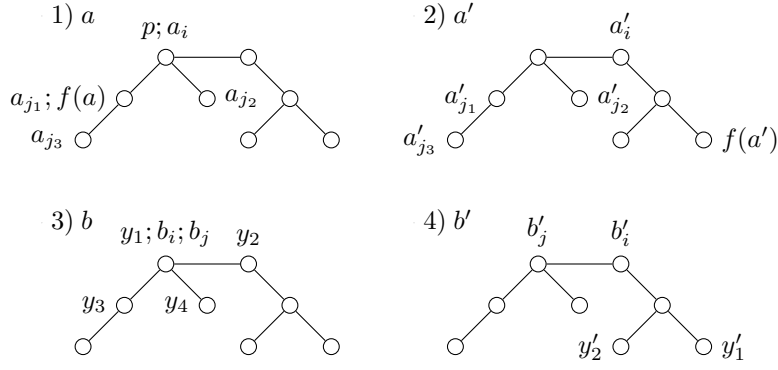
     \subsection{Our Characterization}
    We now complete our characterization of onto, SP mechanisms.
 \begin{lemma}\label{claim:easy}
 Every TMON and ATC mechanism $f$ on the discrete tree is SP.
 \end{lemma}
 \begin{proof}
Suppose $f$ is TMON and ATC. By definition it is 1-TSI and ADR. Consider an arbitrary pair of profiles $(a,a')$ s.t. $a' = (a_{-i}, a'_i)$ for some $i \in N, a'_i \in V$. If $tree(a_i\rightarrow a'_i,f(a)) =  tree(a_i\rightarrow a'_i,f(a'))$, it follows from ADR that
\begin{align}
    &depth(a_i \rightarrow a'_i,f(a)) -  depth( a_i \rightarrow a'_i,f(a')) = 0 \notag\\&=  d(tree(a_i \rightarrow a'_i, f(a)),  tree(a_i \rightarrow a'_i, f(a')))\label{same tree1}
\end{align}
If $tree(a_i\rightarrow a'_i,f(a)) \neq  tree(a_i\rightarrow a'_i,f(a'))$, it follows from 1-TSI that
\begin{align}
d(f(a),[a_i,a'_i]) \leq 1; \textit{ and } d(f(a'),[a_i,a'_i])\leq 1\label{different tree1}
\end{align}
Eq.~\eqref{different tree1} implies that 
\begin{align}
   &|depth(a_i \rightarrow a'_i,f(a)) -  depth( a_i \rightarrow a'_i,f(a_{-i},a'_i))| \leq 1\notag\\&\leq   d(tree(a_i \rightarrow a'_i, f(a)),  tree(a_i \rightarrow a'_i, f(a'))) \notag
\end{align}      
From Eqs.~\eqref{same tree1},~\eqref{different tree1}, $f$ satisfies DB, and therefore $f$ is SP from Lemma~\ref{claim:proposal}.
 \end{proof}

We conclude that every onto mechanism $f$ on the discrete tree is strategyproof if and only if it is TMON and ATC. This follows from Lemmas~\ref{claim:proposal}, \ref{1-DR}, \ref{claim:1-TSI} and \ref{claim:easy}.

\section{Shift-Invariant Mechanisms}\label{sec:linesSI}
  In this section, we show that on infinite discrete lines, the only anonymous, shift-invariant SP mechanisms are order statistics mechanisms---as on continuous lines~\cite{moulin1980strategy}. Therefore, all additional mechanisms that are SP in the discrete domain must single out specific
locations to satisfy Neutrality.
 
 We assume that the vertices of the line are indexed in increasing order and that the agents are ordered by their location in profile $a$, i.e., if $a_i < a_j$, $i < j$. 
\begin{definition}[Shift-Invariant] A mechanism $f$ on an infinite discrete line is \textbf{shift-invariant} if for every location profile $a = (a_1,\ldots,a_n)$ and $d \in N$: $f(a_1 + d,\ldots,a_n + d) = f(a) + d$.
\end{definition}
\begin{definition}[Anonymous] A mechanism $f$ is \textbf{anonymous} if for every location profile $a$ and every permutation of agents $\pi: N \rightarrow N$, it holds that 
$$f(a_1,\ldots,a_n) = f(a_{\pi_1},\ldots,a_{\pi_n})$$
 \end{definition}
 \begin{definition}[$k$th order statistic mechanism] A mechanism $f$ is the $k$th order statistic mechanism for some $k \in N$, if for every profile $a= (a_1,...,a_n)$, it holds that $f(a) = a_k$.
 \end{definition}
\begin{lemma}\label{neutral1}
    If a mechanism $f$ is onto, SP and shift-invariant, then for every profile $a=(a_1,\ldots,a_n)$, it holds that $f(a) = a_i$ for some $i \in [1,n]$.\footnote{This property is sometimes called ``tops-only" or ``peaks-only".}
     \end{lemma}
    \begin{proof}
    Consider the profile $a' = (a_1+1,\ldots,a_n+1)$. From shift-invariance, $f(a') = f(a)+1$. The profile $a'$ can be achieved by $n$ moves of one step, one per an agent. For every such move of agent $j \in N$, resulting in profile $a^j$, the following hold:
    \begin{enumerate}
        \item $f(a^j) \geq f(a)$ from TMON.
        \item If $f(a^j) \neq f(a)$, from shift-invariance and the previous statement, it holds that $f(a^j) = f(a) + 1$.
    \end{enumerate}
    The only way to satisfy these conditions without violating DB, is by demanding that $f(a) = a_i$ for some $i \in N$ to assure that $$tree(a_j \rightarrow a_j+1, f(a)) \neq tree(a_j \rightarrow a_j+1, f(a^j))$$
    \end{proof}
\begin{lemma}\label{neutral2}
    An onto, SP, anonymous, shift-invariant mechanism $f$ for the line is a $k$th order statistic mechanism.
     \end{lemma}
     \begin{proof}
      Assume by contradiction that $f$ is not a $k$th order statistic mechanism, i.e. there exists a pair of profiles $(a,b)$ s.t. $f(a)$ is the $j$th agent location and $f(b)$ is the $k$th agent location. Assume that 
         $j<k$. Let $d$ denote $b_n-a_1 + 1$. Consider a profile $c$ where $c_i = b_i - d$ if $d > 0$ and $c_i = b_i$ otherwise.
    from shift-invariance, it follows that
        $f(c) = c_k$.
        
        \medskip
    We now iteratively move each agent in profile $c$ to a location in profile $a$. In the $l$th iteration we move the agent with index $n+1-l$ to the location $a_{n+1-l}$, getting a sequence of profiles $(c^l)_{l=1}^n$. In every iteration, the output is the location of an agent with an index higher or equal to $k$ from TMON. After the $n$th iteration we reach a profile $c^n$ that is identical to $a$, up to permutation of agents. Thus by anonymity $f(a)=f(c^n)$. 
    
    On the  other hand, $f(c^n)\geq (c^n)_k = a_k > a_j = f(a)$, i.e. a contradiction. The proof is similar for the case $j >k$.
     \end{proof}
\begin{theorem}\label{claim:neutral}
    An onto mechanism $f$ is anonymous, shift-invariant and SP if and only if it is the $k$th order statistic for some $k \in \mathbb{N}$.
    \end{theorem}
     \begin{proof}
      The first direction follows from Lemmas~\ref{neutral1} and \ref{neutral2}.
   
    We prove the second direction. Clearly, a $k$th order statistic mechanism is anonymous and shift-invariant. The only way for an agent $i$ to change the chosen location is by reporting a location $a_i > a_k$ if $i < k$ or a location $a_i < a_k$ if $i > k$. In both cases, the distance of the agent from the facility will increase.
    \end{proof}

   
\section{Conclusion and Open Questions}
\label{chap:conclusion}

In this research, we provide a complete characterization of onto and strategyproof facility location mechanisms on discrete trees---under quadratic preferences. Interestingly, while a characterization for continuous trees exists due to \cite{schummer2002strategy}, these are not easily compared, as the latter uses a collection of median-like rules rather than axiomatic properties. The key property that allows comparison of continuous and discrete mechanisms is \emph{trajectory containment} (TC): while this property characterizes exactly the strategyproof onto mechanisms on continuous trees, it needs to be relaxed in a particular way to apply for discrete trees. 

A different characterization for discrete trees and general single-peaked preferences~\cite{peters2019unanimous} uses a third type of properties that map agents to leafs of the tree. 
     Thus a better understanding of how properties from the different models map onto one another is important.
     
One possible direction for further research is a characterization of strategyproof mechanisms on discrete weighted graphs. Contrary to continuous graphs, there is a finite set of possible locations for each agent and outcome. Unlike the discrete case, the distances between such two neighboring locations vary along the tree.

 Additionally, characterizations of SP mechanisms can promote the study of optimal SP approximation mechanisms, for example for minimizing the Egalitarian cost.

\bibliographystyle{unsrt}  
\bibliography{references}

\end{document}

%% file: alg.tex
\begin{algorithm}[H]
\SetAlgoLined
For every $0 \leq x,y \leq 4$:\\
  \eIf{y=0}{
  $f(x,y) = x \textit{ mod }2$ }
   {\eIf{ y=2}{
   $f(x,y) = min\{x,3\}$}
   {$f(x,y) = min\{x,y\}$} }
\end{algorithm}

%% file: alg_graph.tex
\begin{tikzpicture}
\vertex (a) at (0.2,0.8) [label=above:$a_1$, label=below:$0$] {};
 \vertex (b) at (0.2,-0.8)  [label=below:$1$] {};
  \vertex (k) at (1,0) [label=below:$2$]  {};
 \vertex (c) at (2,0)  [label=above:$a_2$, label=below:$3$] {};  
\vertex (d) at (3,0)  [label=above:$a'_2$, label=below:$4$] {};

\path

(a) edge (k)
(b) edge (k)
(k) edge (c)
(c) edge (d)

;   
\end{tikzpicture}

%% file: double.tex
\begin{tikzpicture}[scale=0.65]
\vertex (h) at (-5,1) [label=above:$a_1$] {};
\vertex (a) at (-4,1)  {};
\vertex (b) at (-3,1) [label=above:$v_1$]{};  
\vertex (c) at (-2,1) [label=above:$a'_1$] {};

\vertex (i) at (-5,0)  [label=left:$z$]{};

\vertex (d) at (-6,-1) [label=above:$v_2$] {};
\vertex (e) at (-4,-1) {};
\vertex (f) at (-7,-2)  [label=left:$f(a)$]{};
\vertex (g) at (-3,-2)  [label=left:$f(a')$]{};
\vertex (j) at (-2,0)  {};
\vertex (k) at (-5,-2)  [label=left:$a_2$] {};

\path
(h) edge (a)
(b) edge (a)
(b) edge (c)
(i) edge (h)
(h) edge (i)
(d) edge (f)
(g) edge (e)
(d) edge (i)
(i) edge (e)
(j) edge (b)
(d) edge (k)
;   
\end{tikzpicture}

%% file: doubleg.tex
\begin{tikzpicture}[scale=0.65]
\vertex (h) at (-5,1) [label=above:$a_1$] {};
\vertex (a) at (-4,1)  [label=above:$g(a')$]{};
\vertex (b) at (-3,1) [label=above:$v_1$]{};  
\vertex (c) at (-2,1) [label=above:$a'_1$] {};

\vertex (i) at (-5,0)  [label=left:$z$]{};

\vertex (d) at (-6,-1) [label=above:$v_2$] {};
\vertex (e) at (-4,-1) {};
\vertex (f) at (-7,-2)  [label=left:$g(a)$]{};
\vertex (g) at (-3,-2) {};
\vertex (j) at (-2,0)  {};
\vertex (k) at (-5,-2)  [label=left:$a_2$] {};

\path
(h) edge (a)
(b) edge (a)
(b) edge (c)
(i) edge (h)
(h) edge (i)
(d) edge (f)
(g) edge (e)
(d) edge (i)
(i) edge (e)
(j) edge (b)
(d) edge (k)
;   
\end{tikzpicture}

%% file: adr_suppl.tex
\begin{tikzpicture}
\vertex (u) at (1,1)  {};
 \vertex (a) at (0,1) [label=above:$a_i$] {};
 \vertex (b) at (3,1)  [label=above:$a'_i$] {};  
\vertex (c) at (2,1)  [label=above:$c_i;z;y$] {};
\vertex (d) at (1,0)  {};
\vertex (e) at (3,0) [label=left:$y'$]  {};
\vertex (f) at (0,-1) [label=left:$f(a)$] {};
\vertex (h) at (4,-1) [label=left:$f(a')$] {};
\vertex (i) at (2,-1) [label=left:$f(c)$] {};
\path
(a) edge (u)
(b) edge (c)
(c) edge (u)
(c) edge (d)
(c) edge (e)
(d) edge (f)
(e) edge (h)
(i) edge (e)

;   
\end{tikzpicture}

%% file: adr.tex
 \begin{tikzpicture}[scale=0.51]
\vertex (u) at (-1,1)  {};
 \vertex (a) at (-2,1)  {};
 \vertex (b) at (-3,1) [label=above:$a_i$] {};  
\vertex (c) at (0,1) [label=above:$z$] {};

\vertex (d) at (-1,0) [label=left:$v_2$] {};
\vertex (e) at (1,0)   {};
\vertex (f) at (-2,-1)  [label=left:$v_1$]{};
\vertex (g) at (0,-1)  {};
\vertex (h) at (2,-1)  {};

\vertex (i) at (-3,-2) [label=left:$a_{j_1};x$] {};
\vertex (j) at (-1,-2) [label=left:$a_{j_2}$] {};

\vertex (q) at (1,-2)  {};
\vertex (r) at (3,-2)  {};

\vertex (s) at (-4,-3) [label=left:$a_{j_3}$] {};
\vertex (qqqqq) at (-5,2) [fill,scale=0.001,label=right:$1)$ $a$]  {};

\vertex (qq) at (9,1)  {};
 \vertex (ww) at (8,1)  {};
 \vertex (ee) at (7,1)  {};  
\vertex (rr) at (10,1) [label=above:$a'_i;z$] {};

\vertex (tt) at (9,0) [label=left:$v_2$] {};
\vertex (yy) at (11,0)   {};
\vertex (uu) at (8,-1)  [label=left:$v_1$]{};
\vertex (ii) at (10,-1)  {};
\vertex (oo) at (12,-1)  {};

\vertex (pp) at (7,-2) [label=left:$a'_{j_1}$] {};
\vertex (qqq) at (9,-2) [label=left:$a'_{j_2}$] {};

\vertex (www) at (11,-2)  {};
\vertex (eee) at (13,-2)[label=right:$x'$]  {};

\vertex (rrr) at (6,-3) [label=left:$a'_{j_3}$] {};
\vertex (wwwww) at (5,2) [fill,scale=0.001,label=right:$2)$ $a'$]  {};

\vertex (aa) at (-1,-6)  {};
 \vertex (ss) at (-2,-6)  {};
 \vertex (dd) at (-3,-6) [label=above:$b_i$] {};  
\vertex (ff) at (0,-6)[label=above:$z$]   {};

\vertex (gg) at (-1,-7) [label=left:$v_2;y_1$] {};
\vertex (hh) at (1,-7)   {};
\vertex (jj) at (-2,-8)  [label=left:$v_1;y_2;b_j$]{};
\vertex (kk) at (1,-8)  {};
\vertex (ll) at (2,-8)  {};

\vertex (aaa) at (-3,-9) [label=left:$x;y_3$] {};
\vertex (sss) at (-1,-9) [label=left:$y_4$] {};

\vertex (ddd) at (1,-9)  {};
\vertex (fff) at (3,-9)[label=right:$x'$]  {};

\vertex (ggg) at (-4,-10)  {};
\vertex (eeeee) at (-5,-5) [fill,scale=0.001,label=right:$3)$ $b$]  {};

\vertex (zz) at (9,-6)  {};
 \vertex (xx) at (8,-6)  {};
 \vertex (cc) at (7,-6) [label=above:$y'_1$] {};  
\vertex (vv) at (10,-6) [label=above:$b'_i;z$] {};

\vertex (bb) at (9,-7) [label=left:$v_2$] {};
\vertex (nn) at (11,-7)   {};
\vertex (mm) at (8,-8)  [label=left:$v_1;b'_j$]{};
\vertex (zzz) at (10,-8)  {};
\vertex (xxx) at (12,-8)  {};

\vertex (ccc) at (7,-9)  {};
\vertex (vvv) at (9,-9) {};

\vertex (bbb) at (11,-9) [label=right:$y'_2$] {};
\vertex (nnn) at (13,-9)[label=right:$y'_3$]  {};

\vertex (mmm) at (6,-10)  {};
\vertex (rrrr) at (5,-5) [fill,scale=0.001,label=right:$4)$ $b'$]  {};

\path

(b) edge (a)
(a) edge (u)
(u) edge (c)
(c) edge (d)
(c) edge (e)
(d) edge (f)
(e) edge (g)
(e) edge (h)
(f) edge (j)
(f) edge (i)
(q) edge (h)
(h) edge (r)
(i) edge (s)

(ww) edge (ee)
(qq) edge (ww)
(qq) edge (rr)
(rr) edge (tt)
(rr) edge (yy)
(tt) edge (uu)
(yy) edge (ii)
(yy) edge (oo)
(uu) edge (pp)
(oo) edge (www)
(pp) edge (rrr)
(uu) edge (qqq)
(eee) edge (oo)

(ss) edge (dd)
(aa) edge (ss)
(aa) edge (ff)
(ff) edge (gg)
(ff) edge (hh)
(gg) edge (jj)
(hh) edge (kk)
(hh) edge (ll)
(jj) edge (aaa)
(ll) edge (ddd)
(aaa) edge (ggg)
(jj) edge (sss)
(fff) edge (ll)

(xx) edge (cc)
(zz) edge (xx)
(zz) edge (vv)
(vv) edge (bb)
(vv) edge (nn)
(bb) edge (mm)
(nn) edge (zzz)
(nn) edge (xxx)
(mm) edge (ccc)
(xxx) edge (bbb)
(ccc) edge (mmm)
(mm) edge (vvv)
(nnn) edge (xxx)
;   
\end{tikzpicture}

%% file: 1-TSI-one-step.tex
\begin{tikzpicture}[scale=0.65]
\vertex (a) at (-1,1) [label=above:$a_i$] {};
 \vertex (b) at (0,1)  {};
  \vertex (k) at (1,1)  {};
 \vertex (c) at (2,1)  [label=above:$x$] {};  
\vertex (d) at (3,1)  [label=above:$y$] {};
\vertex (e) at (4,1)  [label=above:$a'_i$]  {};
\vertex (f) at (-2,0) {};
\vertex (g) at (-3,-1) [label=left:$f(a)$] {};
\vertex (h) at (2,-1) [label=right:$f(y')$] {};
\vertex (i) at (1,0) {};
\vertex (j) at (-1,-1) [label=left:$f(x')$] {};

\path

(b) edge (a)
(b) edge (k)
(k) edge (c)
(d) edge (e)
(c) edge (d)

(a) edge (f)
(f) edge (g)
(b) edge (i)
(i) edge (h)
(f) edge (j)

;   
\end{tikzpicture}

%% file: 1-TSI.tex
\begin{tikzpicture}[scale=0.55]
\vertex (a) at (-8,0) [label=above:$p;a_i$] {};
 \vertex (b) at (-6,0)  {};
 \vertex (c) at (-9,-1) [label=left:$a_{j_1};f(a)$] {};  
\vertex (d) at (-7,-1) [label=right:$a_{j_2}$] {};
\vertex (e) at (-5,-1)  {};
\vertex (f) at (-4,-2)   {};
\vertex (y) at (-6,-2)   {};
\vertex (cc) at (-10,-2) [label=left:$a_{j_3}$]  {};
\vertex (gg) at (-11,1) [fill,scale=0.001,label=right:$1)$ $a$]  {};
\vertex (g) at (1,0)  {};
 \vertex (h) at (3,0) [label=above:$a'_i$] {};
 \vertex (i) at (0,-1) [label=left:$a'_{j_1}$] {};  
\vertex (j) at (2,-1) [label=right:$a'_{j_2}$] {};
\vertex (k) at (4,-1)  {};
\vertex (l) at (5,-2)  [label=right:$f(a')$] {};
\vertex (z) at (3,-2)   {};
\vertex (dd) at (-1,-2)  [label=left:$a'_{j_3}$] {};

\vertex (hh) at (-2,1) [fill,scale=0.001,label=right:$2)$ $a'$]  {};

\vertex (m) at (-8,-5) [label=above:$y_1;b_i;b_j$] {};
 \vertex (n) at (-6,-5)  [label=above:$y_2$] {};
 \vertex (o) at (-9,-6) [label=left:$y_3$] {};  
\vertex (p) at (-7,-6)  [label=left:$y_4$] {};
\vertex (q) at (-5,-6)  {};
\vertex (r) at (-4,-7)   {};
\vertex (aa) at (-6,-7)   {};
\vertex (ee) at (-10,-7)   {};

\vertex (ii) at (-11,-4) [fill,scale=0.001,label=right:$3)$ $b$]  {};

\vertex (s) at (1,-5) [label=above:$b'_j$] {};
 \vertex (t) at (3,-5) [label=above:$b'_i$] {};
 \vertex (u) at (0,-6)  {};  
\vertex (v) at (2,-6)  {};
\vertex (w) at (4,-6)  {};
\vertex (x) at (5,-7)  [label=right:$y'_1$] {};
\vertex (bb) at (3,-7) [label=left:$y'_2$]  {};
\vertex (ff) at (-1,-7)   {};

\vertex (jj) at (-2,-4) [fill,scale=0.001,label=right:$4)$ $b'$]  {};
\path

(b) edge (a)
(a) edge (c)
(a) edge (d)
(b) edge (e)
(f) edge (e)
(y) edge (e)
(c) edge (cc)

(g) edge (h)
(g) edge (i)
(g) edge (j)
(h) edge (k)
(k) edge (l)
(z) edge (k)
(i) edge (dd)

(m) edge (n)
(m) edge (o)
(m) edge (p)
(n) edge (q)
(q) edge (r)
(aa) edge (q)
(o) edge (ee)

(s) edge (t)
(s) edge (u)
(s) edge (v)
(t) edge (w)
(w) edge (x)
(bb) edge (w)
(u) edge (ff)
;   
\end{tikzpicture}

%% file: template.bbl
\begin{thebibliography}{10}

\bibitem{moulin1980strategy}
Herv{\'e} Moulin.
\newblock On strategy-proofness and single peakedness.
\newblock {\em Public Choice}, 35(4):437--455, 1980.

\bibitem{border1983straightforward}
Kim~C Border and James~S Jordan.
\newblock Straightforward elections, unanimity and phantom voters.
\newblock {\em The Review of Economic Studies}, 50(1):153--170, 1983.

\bibitem{dokow2012mechanism}
Elad Dokow, Michal Feldman, Reshef Meir, and Ilan Nehama.
\newblock Mechanism design on discrete lines and cycles.
\newblock In {\em Proceedings of the 13th ACM Conference on Electronic
  Commerce}, pages 423--440, 2012.

\bibitem{schummer2002strategy}
James Schummer and Rakesh~V Vohra.
\newblock Strategy-proof location on a network.
\newblock {\em Journal of Economic Theory}, 104(2):405--428, 2002.

\bibitem{peters2019unanimous}
Hans Peters, Souvik Roy, and Soumyarup Sadhukhan.
\newblock Unanimous and strategy-proof probabilistic rules for single-peaked
  preference profiles on graphs.
\newblock Technical report, Working Paper, 2019.

\bibitem{black1948rationale}
Duncan Black.
\newblock On the rationale of group decision-making.
\newblock {\em Journal of political economy}, 56(1):23--34, 1948.

\bibitem{chen2020strategyproof}
Xin Chen, Qizhi Fang, Wenjing Liu, and Yuan Ding.
\newblock Strategyproof mechanisms for 2-facility location games with minimax
  envy.
\newblock In {\em International Conference on Algorithmic Applications in
  Management}, pages 260--272. Springer, 2020.

\bibitem{liu2020multiple}
Wenjing Liu, Yuan Ding, Xin Chen, Qizhi Fang, and Qingqin Nong.
\newblock Multiple facility location games with envy ratio.
\newblock In {\em International Conference on Algorithmic Applications in
  Management}, pages 248--259. Springer, 2020.

\bibitem{cheng2011mechanisms}
Yukun Cheng, Wei Yu, and Guochuan Zhang.
\newblock Mechanisms for obnoxious facility game on a path.
\newblock In {\em International Conference on Combinatorial Optimization and
  Applications}, pages 262--271. Springer, 2011.

\bibitem{cheng2013strategy}
Yukun Cheng, Wei Yu, and Guochuan Zhang.
\newblock Strategy-proof approximation mechanisms for an obnoxious facility
  game on networks.
\newblock {\em Theoretical Computer Science}, 497:154--163, 2013.

\bibitem{Anastasiadis2018HeterogeneousFL}
Eleftherios Anastasiadis and Argyrios Deligkas.
\newblock Heterogeneous facility location games.
\newblock In {\em The 17th International Conference on Autonomous Agents and
  Multiagent Systems (AAMAS'18)}, 2018.

\bibitem{duan2019heterogeneous}
Lingjie Duan, Bo~Li, Minming Li, and Xinping Xu.
\newblock Heterogeneous two-facility location games with minimum distance
  requirement.
\newblock In {\em The 18th International Conference on Autonomous Agents and
  Multiagent Systems (AAMAS'19)}, pages 1461--1469, 2019.

\bibitem{xu2020strategyproof}
Xinping Xu, Minming Li, and Lingjie Duan.
\newblock Strategyproof mechanisms for activity scheduling.
\newblock In {\em The 19th International Conference on Autonomous Agents and
  Multiagent Systems (AAMAS'20)}, pages 1539--1547, 2020.

\bibitem{todo2011false}
Taiki Todo, Atsushi Iwasaki, and Makoto Yokoo.
\newblock False-name-proof mechanism design without money.
\newblock In {\em The 10th International Conference on Autonomous Agents and
  Multiagent Systems (AAMAS'11)}, pages 651--658, 2011.

\bibitem{nehama2019manipulations}
Ilan Nehama, Taiki Todo, and Makoto Yokoo.
\newblock Manipulations-resistant facility location mechanisms for zv-line
  graphs.
\newblock In {\em Proceedings of the 18th International Conference on
  Autonomous Agents and MultiAgent Systems}, pages 1452--1460, 2019.

\bibitem{todo2019false}
Taiki Todo, Nodoka Okada, and Makoto Yokoo.
\newblock False-name-proof facility location on discrete structures.
\newblock {\em arXiv preprint arXiv:1907.08914}, 2019.

\bibitem{wada2018facility}
Yuho Wada, Tomohiro Ono, Taiki Todo, and Makoto Yokoo.
\newblock Facility location with variable and dynamic populations.
\newblock In {\em The 17th International Conference on Autonomous Agents and
  Multiagent Systems (AAMAS'18)}, pages 336--344, 2018.

\bibitem{procaccia2009approximate}
Ariel~D Procaccia and Moshe Tennenholtz.
\newblock Approximate mechanism design without money.
\newblock In {\em Proceedings of the 10th ACM conference on Electronic
  commerce}, pages 177--186, 2009.

\bibitem{alon2010strategyproof}
Noga Alon, Michal Feldman, Ariel~D Procaccia, and Moshe Tennenholtz.
\newblock Strategyproof approximation of the minimax on networks.
\newblock {\em Mathematics of Operations Research}, 35(3):513--526, 2010.

\bibitem{barbera1990strategy}
Salvador Barbera and Bezalel Peleg.
\newblock Strategy-proof voting schemes with continuous preferences.
\newblock {\em Social choice and welfare}, 7(1):31--38, 1990.

\end{thebibliography}
